\documentclass[11pt,letterpaper]{article}

\usepackage{url}
\usepackage{fullpage}
\usepackage{amsmath}
\usepackage{algorithm}
\usepackage[noend]{algpseudocode}
\usepackage{amssymb}
\usepackage{amsthm}
\usepackage{color}
\usepackage{array}
\usepackage{xy}
\usepackage{setspace}
\usepackage{multicol}
\usepackage{algorithm}
\usepackage{hyperref}
\usepackage[disable]{todonotes}

\newcommand{\dtw}{\operatornamewithlimits{DTW}}

\newtheorem{thm}{Theorem}[section]
\newtheorem{lem}[thm]{Lemma}

\newtheorem{cor}[thm]{Corollary}

\theoremstyle{remark}

\newtheorem{lemma}[thm]{Lemma}

\theoremstyle{remark}

\newcommand{\defn}[1]{\textbf{\emph{#1}}}
\renewcommand{\paragraph}[1]{\vspace{.2 cm} \noindent \textbf{#1}.}
\newcommand{\wt}{\text{wt}}
\newcommand{\aug}{\text{Aug}}
\usepackage{authblk}

\begin{document}

\title{Binary Dynamic Time Warping in Linear Time}
 \author{William Kuszmaul}
\affil{MIT CSAIL \\ \texttt{kuszmaul@mit.edu}}
\date{}
\maketitle

\begin{abstract}
  Dynamic time warping distance (DTW) is a widely used distance measure between time series $x, y \in \Sigma^n$. It was shown by Abboud, Backurs, and Williams that in the \emph{binary case}, where $|\Sigma| = 2$, DTW can be computed in time $O(n^{1.87})$. We improve this running time $O(n)$.

  Moreover, if $x$ and $y$ are run-length encoded, then there is an algorithm running in time $\tilde{O}(k + \ell)$, where $k$ and $\ell$ are the number of runs in $x$ and $y$, respectively. This improves on the previous best bound of $O(k\ell)$ due to Dupont and Marteau.
\end{abstract}
\pagenumbering{gobble}

\newpage
\pagenumbering{arabic} 

\section{Introduction}

Dynamic time warping distance (DTW) is a widely used distance measure
between time series \cite{muller2007dynamic}. DTW is particularly
flexible in dealing with temporal sequences that vary in speed.  To
measure the distance between two sequences, portions of each sequence
are allowed to be warped (meaning that a character may be replaced
with multiple consecutive copies of itself), and then the warped
sequences are compared by summing the distances between corresponding
pairs of characters. DTW's many applications include phone
authentication \cite{dtwapp1}, signature verification \cite{dtwapp2},
speech recognition \cite{dtwapp3}, bioinformatics \cite{dtwapp4},
cardiac medicine \cite{dtwapp5}, and song identification
\cite{dtwapp6}.

The textbook dynamic-programming algorithm for DTW runs in time
$O(n^2)$, which can be prohibitively slow for large inputs. Moreover,
conditional lower bounds \cite{DTWhard, DTWhard2, kuszmaul2019dynamic}
prohibit the existence of a strongly subquadratic-time
algorithm\footnote{An algorithm is said to run in strongly
  subquadratic time if it runs in time $O(n^{2- \epsilon})$ for some
  constant $\epsilon > 0$. Although strongly subquadratic time
  algorithms are prohibited by conditional lower bounds, runtime
  improvements by subpolynomial factors are not. Such improvements
  have been achieved \cite{DTWsubquadratic}.}, unless the Strong
Exponential Time Hypothesis is false.

On the practical side, the difficulty of computing DTW directly has
motivated the development of fast heuristics
\cite{dtwband,kp99,kp00,k02,BUWK15,PFWNCK16} which typically lack
provable guarantees.

On the theoretical side, the difficulty of computing DTW directly has
led researchers to focus on certain important special cases.\footnote{Researchers have also studied related problems that are not constrained by the aforementioned conditional lower bounds. See, for example, work by Braverman et al. \cite{comm} on communication complexity and by Kuszmaul \cite{kuszmaul2019dynamic} on approximation algorithms.}  Hwang
and Gelfand \cite{hwang2017sparse} show how to compute $\dtw(x, y)$ in time
$O((s + t) n)$, where $|x| = |y| = n$ and where $s$ and $t$ are the
number of non-zero values in $x$ and $y$, respectively. Kuszmaul
\cite{kuszmaul2019dynamic} showed how to compute $\dtw(x, y)$ in time
$O(n \dtw(x, y))$, and also gave an $O(n^{\epsilon})$-approximation
algorithm with running time $\tilde{O}(n^{2 - \epsilon})$. Recently,
Froese et al. \cite{froese2019fast} gave an algorithm parameterized by the
run-length-encoding lengths of $x$ and $y$, running in time
$O((k + \ell) n)$, where $k$ and $\ell$ are the number of
repeated-letter runs in $x$ and $y$ respectively. In the case where
$k, \ell \in O(\sqrt{n})$, the algorithm achieves a faster time of
$O(k^2 \ell + \ell^2 k)$.

\paragraph{Binary DTW} One case that is of special interest is that
where $x$ and $y$ are \defn{binary time series} -- that is,
$x, y \in \{0, 1\}^n$. In this case, the conditional lower bounds
\cite{DTWhard, DTWhard2, kuszmaul2019dynamic} do not apply. Abboud,
Backurs, and Williams \cite{DTWhard2} gave an algorithm for computing
binary DTW in time $O(n^{1.87})$, building on an algorithm given by
\cite{convolution} for the Bounded Monotone Convolution Problem\footnote{For a full discussion of the $O(n^{1.87})$-time algorithm, see the extended version \cite{DTWhard2a} of \cite{DTWhard2}.}. Other
work has given algorithms running in time $O(st)$
\cite{hwang2019binary, mueen2016awarp}, where $s$ and $t$ are the
number of $1$s in $x$ and $y$ respectively, and in time $O(k \ell)$,
where $k$ and $\ell$ are the number of repeated-letter runs in $x$ and
$y$ respectively \cite{dupont2015coarse}.

The binary DTW problem has also received attention from practitioners. For example,
several of the CASAS human activity data sets \cite{casas} that have been
examined in the context of DTW \cite{mueen2016awarp, schaar2020faster}
consist of binary data points (e.g., sensor data indicating when a
door is open/closed).

Binary DTW has also been studied in the context of a large
number $r$ of time series $x^{(1)}, x^{(2)}, \ldots, x^{(r)}$ being
considered simultaneously. In this case, researchers have focused on
the Binary Mean Problem \cite{schaar2020faster}, in which the goal is
to find a single time series $x^*$ that minimizes the sum of dynamic
time warping distances $\sum_i \dtw(x^*, x^{(i)})$. Leveraging the
binary DTW algorithm of \cite{DTWhard2}, Schaar, Froese, and
Niedermeier \cite{schaar2020faster} gave an $O(r n^{1.87})$-time
algorithm for the Binary Mean Problem. The algorithm was not included
in the subsequent empirical evaluation \cite{schaar2020faster},
however, due to the impracticality of the $n^{1.87}$ term.

\paragraph{Binary DTW in linear time}
In this note, we show that binary DTW can be computed in linear time
$O(n)$, substantially improving on the previous state of the art of
$O(n^{1.87})$. Our algorithm is very simple, and hinges on the
relationship between binary DTW and minimum weight bipartite matching.

Our algorithm can also be modified for the case where $x$ and $y$ are
run-length encoded. If $x$ and $y$ consist of $k$ and $\ell$
repeated-letter runs, respectively, then our algorithm runs in time
$O((k + \ell) \log (k + \ell))$. 

\paragraph{An alternative solution using the Monge property}
After writing this paper, we also learned of an alternative solution to computing binary DTW in (near) linear time. As we shall discuss in more detail later, Abboud, Backurs, and Williams \cite{AbboudBa15} reduce the problem of binary DTW to the following: given a sequence of $I$ numbers $a_1, \ldots, a_I$ and a value $r$, choose a subsequence of $r$ elements that minimizes the sum, subject to the constraint that no two elements are adjacent. In subsequent work on knapsack and graph algorithms, Axiotis and Tzamos \cite{ax} give a $O(n \log n)$-time solution to this problem (Lemma 19 of \cite{ax}) using results for how to solve a general class of dynamic programs with the so-called Monge property. Thus, our work serves two main purposes: to explicitly make the observation that binary DTW can be computed quickly, and to give an extremely simple algorithm that achieves truly linear time. 

\section{Preliminaries}\label{secpreliminaries}

In this paper we capture treat time series as strings. The \defn{runs}
of a string are the maximal substrings consisting of a single repeated
letter. For example, the runs of $aabbbccd$ are $aa$, $bbb$, $cc$, and
$d$. Given a string $x$, we can \defn{extend} a run in $x$ by further
duplicating the letter which populates the run. For example, the
second run in $aabbbccd$ can be extended to obtain $aabbbbccd$. Any
string obtained from $x$ by extending $x$'s runs is an
\defn{expansion} of $x$. For example, $aaaabbbbccccdddd$ is an
expansion of $aabbbccd$.

Consider two strings $x$ and $y$ with characters from a metric space
$(\Sigma, d)$. A \defn{correspondence} between $x$ and $y$ is a pair
$(\overline{x}, \overline{y})$ of equal-length expansions of $x$ and
$y$. The \defn{value} of a correspondence is the difference
$$\sum_i d(\overline{x}_i, \overline{y}_i)$$ between the two
expansions. A correspondence between $x$ and $y$ is said to be
\defn{optimal} if it has the minimum attainable value, and the
resulting value is called the \defn{dynamic time warping distance}
$\dtw(x, y)$ between $x$ and $y$.

\section{Computing Binary DTW in Linear Time}

In this section we show that binary DTW can be computed in linear time:
\begin{thm}
  Let $x \in \{0, 1\}^n$ and $y \in \{0, 1\}^m$ be binary
  strings. Then $\dtw(x, y)$ can be computed in time $O(n + m)$.
  \label{thm:main}
\end{thm}

We also consider the case where $x$ and $y$ are run-length
encoded. That is, $x$ (and similarly $y$) is given as a sequence of
pairs $(i_1, a_1), (i_2, a_2), \ldots$ indicating that the $j$-th run
consists of $i_j$ copies of the letter $a_j$.

\begin{thm}
  Let $x \in \{0, 1\}^n$ and $y \in \{0, 1\}^m$ be binary
  strings. Suppose that $x$ and $y$ are run-length encoded, and that
  the total number of runs in $x$ and $y$ is $\ell$.  Then
  $\dtw(x, y)$ can be computed in time $O(\ell \log \ell)$.
\label{thm:main2}
\end{thm}

\paragraph{A useful reduction}
We begin by employing a result of Abboud, Backurs, and
Williams \cite{DTWhard2}.

\begin{lem}[Theorem 8 of \cite{DTWhard2}]
  Computing $\dtw(x, y)$ of two strings $x \in \{0, 1\}^n$ and
  $y \in \{0, 1\}^m$ can be reduced in time $O(m + n)$ to (a constant
  number of instances of) the following problem: given a sequence
  $w_1, w_2, \ldots, w_s$ of $s \le \max(m, n)$ positive integers, and
  an integer $r \le s$, find a subsequence
  $w_{i_1}, w_{i_2}, \ldots, w_{i_r}$ of length $r$ that does not use
  any neighboring integers (i.e., $i_{j + 1} - i_j \ge 2$ for all $j$)
  and such that the sum $\sum_j w_{i_j}$ of integers is minimized. The
  integers in $w_1, w_2, \ldots, w_s$ sum up to at most $\max(m, n)$.
  \label{lem:reduction}
\end{lem}

When $x$ and $y$ are run-length encoded (meaning each run is encoded
by its length), then the following extension of Lemma
\ref{lem:reduction} is also useful.
\begin{cor}
  Suppose that $x$ and $y$ are run-length encoded, and that $k$ and
  $\ell$ are the number of runs in $x$ and $y$, respectively. Then the
  reduction in Lemma \ref{lem:reduction} takes time $O(k + \ell)$ and
  results in sequences $w_1, w_2, \ldots, w_s$ of length
  $s \le \min(k, \ell)$.
  \label{cor:reduction}
\end{cor}

Although we will not re-prove Lemma \ref{lem:reduction} here, we do
give a brief intuition. Suppose for simplicity that both $x$ and $y$
begin and end with $0$, and suppose that $x$ has more runs than
$y$. If $x$ has $k$ runs and $y$ has $\ell$ runs, then the optimal
correspondence $(\overline{x}, \overline{y})$ will select $k - \ell$
runs $R$ in $x$ and the correspondence will contain miss-matches
$\overline{x}_i \neq \overline{y}_i$ only for $\overline{x}_i$'s from
those runs $R$. In particular, the expansion $\overline{y}$ of $y$
``covers up'' the runs in $R$ by expanding runs in $y$ to engulf the
runs in $R$. A run in $y$ can only ``cover up'' a run in $x$ if the
two runs have different values (one run is of $0$s and the other is of
$1$s). Consequently, the runs $R$ in $x$ that are covered up cannot be
adjacent to one-another. That is, no two runs in $R$ can appear
adjacently in $x$. This turns out to be the only constraint on $R$,
however, and subject to this constraint, the cost of the
correspondence $(\overline{x}, \overline{y})$ is minimized by
selecting the runs in $R$ to have the minimum possible total
length. Thus the reduction from Lemma \ref{lem:reduction} can be
thought of as follows: let $w_1, w_2, \ldots, w_k$ be the lengths of
the runs in $x$. Then the dynamic time warping distance is given by
$$\dtw(x, y) = \min_{i_1, i_2, \ldots, i_{k - \ell}} \sum_{j = 1}^{k - \ell} w_{i_j},$$
where $i_1 < i_2 < \cdots < i_{k - \ell}$ and where
$i_j + 1 \neq i_{j + 1}$ for any $j$. In order to handle cases where
$x$ and $y$ disagree in their first or last letters, a small amount of
additional casework is necessary, resulting in a reduction to $O(1)$
instances of the subsequence problem, rather than just a single
instance \cite{DTWhard2}.


\paragraph{Relationship to bipartite matching}
The problem given by Lemma \ref{lem:reduction} can be reformulated as
a problem of minimum-weight bipartite matching. Consider the line
graph $G$ with vertices $V = \{v_0, v_1, \ldots, v_s\}$, with edges
$E = \{e_1 = (v_0, v_1), \ldots, e_s = (v_{s - 1}, v_s)\}$, and with
edge-weights $\wt(e_i) = w_i$. Then the problem described in Lemma
\ref{lem:reduction} becomes: find the minimum-weight matching
$M \subseteq E$ such that $|M| = r$.

Our algorithm for computing $\dtw(x, y)$ hinges on the relationship to
minimum-weight bipartite matching. In order to efficiently compute
$\dtw(x, y)$, we will construct the minimum-weight matching $M$ of size
$|M| = r$ by simply performing iterative path augmentation.

\paragraph{The Hungarian Algorithm for weighted bipartite matching}
One of the simplest algorithms for weighted bipartite matching is the
so-called \defn{Hungarian Algorithm} \cite{hungarian1, hungarian2,
  hungarian3, thorup2004integer,fredman1987fibonacci}. Although the
Hungarian Algorithm applies to arbitrary weighted bipartite graphs, we
will be discussing the algorithm and its properties exclusively in the
context of our line graph $G$. In order to describe the algorithm in
the context of a line graph, we first introduce several useful
notations.

Formally, a \defn{matching} $M$ in the line graph $G$ is a subset
$M \subseteq E$ such that $|M \cap \{e_i, e_{i + 1}\}| \le 1$ for each
$i$. The \defn{weight} $\wt(M)$ is given by $\sum_{e_i \in M} \wt(e_i)$. A
\defn{chain} $C$ in a matching $M$ is a set of the form
$C = \{e_i, e_{i + 2}, e_{i + 4}, \ldots, e_{i + 2c}\} \subseteq M$
for some $c \in \mathbb{N}$. The chain $C$ is \defn{maximal} if
$e_{i - 2}, e_{i + 2c + 2} \not\in M$. The \defn{augmentation} of a
chain $C = \{e_i, e_{i + 2}, e_{i + 4}, \ldots, e_{i + 2c}\}$ is the
new chain
$\aug(C) = \{e_{i - 1}, e_{i +1}, e_{i + 3}, \ldots, e_{i + 2c + 1}\}
\cap E$. A matching $M'$ is said to be a \defn{augmentation} of a
matching $M$ if $M' = M \setminus C \cup \aug(C)$ for some maximal
chain $C$ in $M$, and if $|M'| = |M| + 1$.\footnote{Note that $M \setminus C \cup \aug(C)$ evaluates as $(M \setminus C) \cup \aug(C)$ by order of operations.} Note that
$M \setminus C \cup \aug(C)$ is guaranteed to be a matching for any
maximal chain $C$.

In order to simplify discussion, we also introduce the notion of an
empty chain. For $i \in [s]$, a matching $M$ contains the \defn{\boldmath$i$-th
  empty chain} $\emptyset_i$ if $M$ does not contain any of
$e_{i - 1}, e_i, e_{i + 1}$. In this case the empty chain
$\emptyset_i$ is considered to be maximal, and the augmentation
$\aug(\emptyset_i)$ is defined to be $e_i$. Thus, if a matching $M'$
equals $M \cup e_i$ for some edge $e_i \not\in M$, then the matching
$M'$ can be thought of as
$M \setminus \emptyset_i \cup \aug(\emptyset_i)$, making $M'$ an
augmentation of $M$.

The Hungarian Algorithm constructs a matching $M$ of size $r$ as
follows. The algorithm begins with the empty matching $M_0$. The
algorithm then iteratively constructs $M_1, M_2, \ldots, M_r$, where
each $M_i$ is a minimum-weight augmentation of $M_{i - 1}$. That is,
$M_i$ is permitted to be any augmentation of $M_{i - 1}$ that achieves
the minimum attainable value for $\wt(M_{i})$ (over all augmentations of $M_{i - 1}$). The final matching
$M_r$ consists of $r$ edges and is given as the output matching $M$.

Tarjan and Ramshaw (Proposition 3-8 of \cite{ramshaw2012minimum})
showed that the Hungarian Algorithm outputs a matching $M_r$ with the
minimum possible weight (out of all $r$-edge matchings). Note that we focus only on $r \le \lceil s / 2 \rceil$, since $\lceil s / 2 \rceil$ is the size of the largest matching in our line graph $G$.

\begin{lemma}[Proposition 3--8 of \cite{ramshaw2012minimum}]
  For $r \le \lceil s / 2 \rceil$, the matching $M_r$ has the minimum
  weight out of all $r$-edge matchings.
  \label{lem:hungarian}
\end{lemma}

Whereas Tarjan and Ramshaw extend Lemma \ref{lem:hungarian} to
arbitrary bipartite graphs, we are only interested in the line
graph. This allows for an especially simple proof of the lemma.

\begin{proof}[Proof of Lemma \ref{lem:hungarian}]
  Let $r \ge 1$ and suppose by induction that $M_{r - 1}$ is
  minimum-weight out of $(r - 1)$-edge matchings.

  Let $M^*_r$ be a minimum-weight matching of size $r$. The edges $E$
  can be decomposed as the disjoint union,
  \begin{equation}
    E = \bigcup_{\text{maximal chain }C \subseteq M_{r - 1}} C \cup \aug(C).
    \label{eq:decomp}
  \end{equation}
  Since $|M^*_r| = |M_{r - 1}| + 1$, and since $E$ decomposes into
  \eqref{eq:decomp}, there must be a maximal chain
  $C$ in $M_{r - 1}$ for which
  $$|M^*_{r} \cap (C \cup \aug(C))| > |M_{r - 1} \cap (C \cup
  \aug(C))| = |C|.$$ Recalling that $M^*_r$ is a matching, it follows
  that $M^*_r$ contains the chain $\aug(C)$ of size $|C| + 1$.

  Now we turn our attention to $M_r$, the minimum-weight augmentation
  of $M_{r - 1}$. Using the definition of $M_r$, we know that
  $\wt(M_r) \le \wt(M_{r - 1} \setminus C \cup \aug(C))$. To prove that $M_r$ is optimal out of $r$-edge matchings, it therefore suffices to show that
  \begin{equation}
    \wt(M_{r - 1} \setminus C \cup \aug(C)) \le \wt(M^*_r).
    \label{eq:newwt}
  \end{equation}
  By the assumption that $M_{r - 1}$ is a minimum-weight matching, we
  know that
  \begin{equation}
    \wt(M_{r - 1}) \le \wt(M^*_r \setminus \aug(C) \cup C).
    \label{eq:oldwt}
  \end{equation}
  If we remove $C$ from the matchings on both sides of
  \eqref{eq:oldwt}, and then insert $\aug(C)$ into both matchings,
  then we arrive at \eqref{eq:newwt}, as desired.
\end{proof}

\paragraph{Efficiently constructing the matchings}
Again using the fact that $G$ is a line graph on $s + 1$ vertices, the
matchings $M_0, M_1, M_2, \ldots$ can easily be computed in time
$O(s \log s)$. 

\begin{lemma}
  For any $r \le \lceil s / 2 \rceil$, the matching $M_r$ can be
  computed in time $O(s \log s)$.
  \label{lem:run_length}
\end{lemma}
\begin{proof}
  We build $M_0, M_1, \ldots, M_r$ using the Hungarian algorithm. When
  going from $M_i$ to $M_{i + 1}$, we maintain two data structures:
  (1) a balanced binary tree $\mathcal{B}$ consisting of the maximal chains
  $C \subseteq M_i$ for which $|\aug(C)| = |C| + 1$, and sorted by the
  key $\wt(\aug(C)) - \wt(C)$; and (2) an array $\mathcal{A}$ of $s$
  ones and zeroes, where the ones correspond to the positions in which
  the maximal chains $C \subseteq M_i$ begin and end.
  
  To go from $M_i$ to $M_{i + 1}$, the minimum element of
  $\mathcal{B}$ is used to determine which chain $C$ to augment. This
  means that $M_{i + 1} = M_i \setminus C \cup \aug(C)$.  The array
  $\mathcal{A}$ is updated to reflect the update from $M_i$ to
  $M_{i + 1}$, and is used to determine whether the new augmented
  chain $\aug(C)$ combines with another chain $C' \subseteq M_i$ in
  order to form a larger maximal chain in $M_{i + 1}$. The tree
  $\mathcal{B}$ is then updated appropriately to reflect the transition from $M_i$ to $M_{i + 1}$.
  (The subtle case here is that, if $\aug(C)$ combines with another chain $C'$, then
  both $C$ and $C'$ are removed from $\mathcal{B}$ and replaced with a
  single node for the new chain $\aug(C) \cup C'$.)

  The tree $\mathcal{B}$ takes time $O(s \log s)$ to initialize and
  the array $\mathcal{A}$ takes time $O(s)$ to initialize (as all
  zeros). Constructing $M_r$ then takes time $O(r \log s)$.
\end{proof}

Corollary \ref{cor:reduction} and Lemma \ref{lem:run_length} combine
to imply Theorem \ref{thm:main2}.

In order to prove Theorem \ref{thm:main} we will need to prove several
additional properties of the matchings $M_0, M_1, M_2, \ldots$. The next lemma
shows that $M_{i + 2}$ can always be reached from $M_i$ via two
\emph{disjoint} chain augmentations.

\begin{lem}
  Consider $M_i$ and $M_{i + 2}$ for some $i$
  (satisfying $0 \le i \le \lceil s / 2 \rceil - 2$). There exist
  maximal chains $C_1$ and $C_2$ in $M_i$ such that
  $$M_{i + 2} = M_i \setminus (C_1 \cup C_2) \cup (\aug(C_1) \cup \aug(C_2)).$$
  \label{lem:double_jump}
\end{lem}
\begin{proof}
  Let $D_1$ be the maximal chain augmented between $M_i$ and
  $M_{i + 1}$, and let $D_2$ be the maximal chain augmented between
  $M_{i + 1}$ and $M_{i + 2}$. If $D_2$ is a maximal chain in $M_i$,
  then we can simply set $C_1 = D_1$ and $C_2 = D_2$ in order to
  complete the lemma. On the other hand, if $D_2$ is not a maximal
  chain in $M_i$, then $D_2$ must be of the form $D_2' \cup \aug(D_1)$
  for some maximal chain $D_2'$ in $M_i$. It follows that
  $$M_{i + 2} = M_i \setminus (D_1 \cup D_2') \cup \aug(\aug(D_1)) \cup \aug(D_2').$$
  Observe $\aug(\aug(D_1))$ overlaps $\aug(D_2')$ in one edge, and
  otherwise consists of $D_1$ and some other new edge $e_j$. That is, $\aug(\aug(D_1)) \setminus (D_1 \cup \aug(D_2'))$ consists of a single edge $e_j$. It
  follows that
  $$M_{i + 2} = M_i \setminus (\emptyset_j \cup D_2') \cup (\aug(\emptyset_j) \cup \aug(D_2')),$$
  where $\emptyset_j$ is treated as the empty set.
  Setting $C_1 = \emptyset_j$ and $C_2 = D_2'$ completes the proof.
\end{proof}

Using Lemma \ref{lem:double_jump}, we can prove a monotonicity
property for $\Delta_1, \Delta_2, \ldots, \Delta_r$, where
$\Delta_i = \wt(M_i) - \wt(M_{i - 1})$.

\begin{lem}
  Let $r \le \lceil s / 2 \rceil$. Define
  $\Delta_1, \Delta_2, \ldots, \Delta_r$, where
  $\Delta_i = \wt(M_i) - \wt(M_{i - 1})$. Then
  $$\Delta_1 \le \Delta_2 \le \cdots \le \Delta_r.$$
  \label{lem:monotonicity}
\end{lem}
\begin{proof}
  To compare $\Delta_i$ and $\Delta_{i + 1}$, we apply Lemma
  \ref{lem:double_jump} deduce that
  $$M_{i + 1} = M_{i - 1} \setminus (C_1 \cup C_2) \cup (\aug(C_1) \cup \aug(C_2)),$$
  for some two maximal chains $C_1, C_2 \subseteq M_{i - 1}$ (such
  that $|\aug(C_j)| \ge |C_j|$ for both chains $j \in \{1, 2\}$). It follows that
  $$\Delta_i + \Delta_{i + 1} = \wt(C_1) - \wt(\aug(C_1)) + \wt(C_2) - \wt(\aug(C_2)).$$
  This means that for some $j \in \{1, 2\}$, we have
  $$\frac{\Delta_i + \Delta_{i + 1}}{2} \ge \wt(C_j) - \wt(\aug(C_j)).$$
  By the definition of $\Delta_i$, and the fact that $M_i$ is the minimum-weight augmentation of $M_{i - 1}$,
  \begin{align*}
    \wt(M_{i - 1}) + \Delta_i    & = \wt(M_i) \\
                                 & \le \wt(M_{i - 1}) - \wt(C_j) + \wt(\aug(C_j)) \\
                                 & \le \wt(M_{i - 1}) + \frac{\Delta_i + \Delta_{i + 1}}{2}.
  \end{align*}
  It follows that $\Delta_i \le (\Delta_i + \Delta_{i + 1}) / 2$,
  which implies $\Delta_i \le \Delta_{i + 1}$.
\end{proof}

By exploiting the monotonicity of the $\Delta_i$'s, we can compute the
matching $M_r$ in time $O(n + m)$.

\begin{lem}
  For any $r \le \lceil s / 2 \rceil$, the matching $M_r$ can be
  computed in time $O(m + n)$.
  \label{lem:linear_time}  
\end{lem}
\begin{proof}
  We modify the approach from Lemma \ref{lem:run_length} as
  follows. Rather than maintaining $\mathcal{B}$ as a balanced binary
  tree, we maintain $\mathcal{B}$ using what is essentially a dynamic
  bucket sort.

  At any given moment, $\mathcal{B}$ consists of $\max(n, m)$ buckets,
  where each bucket $i$ contains a linked list of the maximal chains
  $C$ whose key $\wt(\aug(C)) - \wt(C)$ equals $i$. (We also modify
  $\mathcal{A}$ to contain a pointer from the one-entries that
  represent the ends of chain $C$ to the linked-list element for $C$
  in $\mathcal{B}$.) Additionally, $\mathcal{B}$ maintains a counter
  $t$ indicating $\Delta_i$ for the most recent $M_i$ computed. In
  order to find the smallest element of $\mathcal{B}$, one simply
  repeatedly increments the counter $t$ until reaching a non-empty
  bucket, and then uses a chain $C$ from that bucket. By Lemma
  \ref{lem:monotonicity}, this always results in us finding the chain
  in $\mathcal{B}$ with the smallest key (i.e., there are never any
  non-empty buckets with indices smaller than our counter $t$).

  The initial state of $\mathcal{B}$ can be constructed time
  $O(m + n)$, since we are inserting $s \le m + n$ elements into
  buckets. The counter $t$ can only be incremented a total of
  $\max(m, n)$ times, and besides those increments, each operation on
  $\mathcal{B}$ takes constant time (making $O(1)$ modifications to
  linked lists).  It follows that the total running time of the
  algorithm is now $O(m + n)$, as desired.
\end{proof}

Lemma \ref{lem:reduction} and Lemma \ref{lem:linear_time} combine to
imply Theorem \ref{thm:main}.

\section{Conclusion}

This note gives a very simple linear time algorithm that computes
$\dtw(x, y)$ for two binary time series $x, y$. The algorithm makes
use of a simple connection between dynamic time warping and
minimum-weight bipartite matching. Although both the algorithm and the
analysis are extremely simple, the linear running time significantly
improves on the previous state of the art of $O(n^{1.87})$~\cite{DTWhard2}.

\paragraph{An open question}
Many applications of dynamic time warping use a \emph{constrained}
version of DTW, in which two the expansions $\overline{x}$ and
$\overline{y}$ are only allowed to pair up letters $x_i$ and $y_j$ if
$|i - j| \le k$ for some width-parameter $k$. This heuristic is known
as the Sakoe-Chiba Band heuristic \cite{dtwband} and is employed, for
example, in the commonly used library of Giorgino \cite{giorgino}. One
of the main reasons that the $k$-width constraint is added is that it
allows for a simple $O(nk)$-time algorithm (which is much faster than
$O(n^2)$ for small $k$). On the other hand, in the case of binary DTW,
the $k$-width constraint may also make DTW a richer similarity
measure. In particular, without the width constraint $\dtw(x, y)$
depends only on the number of runs in $x$ and $y$, and on the
properties of the string with more runs.

Thus we conclude with the following open question. What is the fastest
that binary DTW can be computed subject to the $k$-width constraint?
And, in particular, do $O(m + n)$-time algorithms exist for all $k$?

\section{Acknowledgments}
I would like to thank an anonymous reviewer for pointing out the related work of \cite{ax}.

\vspace{1 cm}


\bibliographystyle{plain}
\bibliography{writeup}

\end{document}